\documentclass[12pt]{article}%
\usepackage[a4paper, total={6.2in, 8.2in}]{geometry}
\usepackage{amsmath}
\usepackage{graphicx}
\usepackage{amsfonts}
\usepackage{amssymb}
\usepackage{color}%

\usepackage{rotating}
\usepackage{tikz}

\newtheorem{theorem}{Theorem}

\usepackage{graphicx}
\newtheorem{lemma}{Lemma}

\newtheorem{remark}{Remark}

\newenvironment{proof}[1][Proof: ]{\textbf{#1}\it }

\begin{document}
\title{Functional principal component analysis for functional data with detection limits 
}
\author{Haiyan Liu$^{1}$
and Jeanine Houwing-Duistermaat$^{1,2}$\\
$^1$Department of Statistics, University of Leeds, United Kingdom\\
$^2$Department of Mathematics,\\ Radboud University, Nijmegen, The Netherlands\\
}

\maketitle

\begin{abstract}
When measurements fall below or above a detection threshold, the resulting data are missing not at random (MNAR), posing challenges for statistical analysis. For example, in longitudinal biomarker studies, observations may be subject to detection limits. Functional principal component analysis (FPCA) is commonly used method for dimension reduction of dense and sparse data measured along a continuum, but standard approaches typically ignore MNAR mechanisms by imputing detection limit values, leading to biased estimates of principal components and scores.

Building on recent work by Liu and Houwing-Duistermaat (2022, 2023), who proposed estimators for the mean and covariance functions under detection limits, we extend FPCA to accommodate functional data affected by such limits. We derive the asymptotic properties of the resulting estimators and assess their performance through simulations, comparing them to standard methods. Finally, we illustrate our approach using longitudinal biomarker data subject to detection limits.

Our method yields more accurate estimates of functional principal components and scores, enhancing the reliability of functional data analysis in the presence of detection limits.

\end{abstract}


\section{Introduction}
Technological advances have led to an increase in datasets with temporal observations, both dense and sparse, prompting the development of functional data analysis (FDA) methods. See for example Ramsay and Silverman (2005), Ferraty and Vieu (2006), Horv\'ath and Kokoszka (2012), Kokoszka and Reimherr (2017) for dense data and Yao et al. (2005), Peng and Paul (2009), Li and Hsing (2010), Wang et al. (2016), Zhang and Wang (2016) for sparse data. While existing approaches assume fully observed data, many real-world datasets, such as longitudinal biomarker measurements, are subject to detection limits (DL), where values below (or sometimes above) a threshold cannot be accurately measured. This results in data that are missing not at random, often appearing as zeros or fixed threshold values. Although recent work has addressed estimation of the mean and covariance functions under such constraints (Shi et al.\ 2021, Liu and Houwing-Duistermaat 2022, 2023), methods for estimating functional principal component scores and eigenfunctions remain underdeveloped. This paper addresses that gap by focusing on FPCA for data with lower DLs, where values below a known threshold are unobservable and replaced by the DL.

FPCA approaches are often applied in studies aiming to 
investigate the effect of a covariate measures along a continuum on a response variable. For example modelling the relationship between fMRI and risk taking behaviour or the effect of a longitudinal measured biomarker on disease progression. The Karhunen-Loeve expansion and Mercer's theorem enables writing the functional predictor in an infinite sum of FPC scores and eigenfunctions. 
Dimension reduction is achieved by truncating the infinite expansion to a finite number of terms, selected based on the proportion of variance they explain. These finite sums can then be included in a regression model for an outcome variable. 

For estimation of the mean function,  Shi et al. (2021) proposed a global method  when there are detection limits. However, since  observations close to the target point $t$ contain more information about the mean function at $t$ than  observations   far away from $t$, local methods as proposed by Liu and Houwing-Duistermaat (2022, 2023) might be more appropriate. 
The estimators for the mean and the covariance function proposed by Liu and Houwing-Duistermaat (2022, 2023) are based on approximations of the likelihood function, while Murphy et al (2024) used a full likelihood approach. For data subject to DL, the likelihood function is a product of probability density functions for the observed values and of probability distribution functions for the observations subject to DL, since for the latter observations we know that the unobserved value is below a known threshold. 
To estimate the mean and covariance functions around observed time points, Liu and Houwing-Duistermaat (2022, 2023) proposed to use the local polynomial kernel method (Fan and Gijbels 1995, Fan and Gijbels 2018, Beran and Liu 2014, Beran and Liu 2016). To reduce the computation time, Liu and Houwing-Duistermaat (2022) developed constant approximations for the probability distribution functions in the likelihood function. The constant approximation is computationally fast especially for dense data while it only performs slightly less than the exact and linear approximation method (Liu and Houwing-Duistermaat, 2022).  Therefore, Liu and Houwing-Duistermaat (2023) used the constant approximations for the covariance estimator and we will use the constant approximation in this paper as well.

Further, two weighting schemes for subjects have been considered (Zhang and Wang 2016, Liu and Houwing-Duistermaat 2022, 2023), namely the SUBJ scheme which assigns the same weight to each subject and the OBS scheme which assigns the same weight to each observation. The latter scheme will assign more weight to subjects with more observations. 
 Liu and Houwing-Duistermaat (2022, 2023) showed that, OBS scheme is similar to SUBJ scheme in estimating mean and covariace under the criteria of MISE for the dense case, while OBS scheme is better than SUBJ scheme in the sparse case. 
Therefore, in this paper, we only consider OBS scheme in the estimation of functional principal components and their corresponding scores.

In this paper, we propose a novel estimator for the functional principal components based on local constant approximations. We first obtain estimates of the eigenfunctions by using the estimator of the covariance function developed for data subject to DL (Liu and Houwing-Duistermaat, 2023). In the next step, we estimate the subject specific scores by using a maximum likelihood approach with constant approximations.  We derive the asymptotic behaviour of the estimators of the scores. 
Via simulations we evaluate their performance in a sparse and a dense setting
and compare their performance with the standard method where the detection limit is used for the missing values. We also investigate the asymptotic behaviour of the estimators via simulations. 
To illustrate the proposed method, we apply it to temporal data from a biomarker study. We finish with a conclusion.

\section{Methodology}
\subsection{Functional Principal Component Analysis (FPCA)}
Firstly, we define the model for functional data subject to a detection limit.
Let $\{X(t): t\in I\}$ represent an $L^2$ stochastic process on the interval $I$.
The mean and covariance function of $X(t)$ are denoted as $\mu(t)=E[X(t)]$ and $C(s,t)=E[(X(s)-\mu(s))(X(t)-\mu(t))]$, respectively. 
Then the process $X(t)$ can be written as:
$$X(t)=\mu(t)+U(t)$$
where $U(t)$ represents the stochastic part of $X(t)$ with mean zero, i.e. $E[U(t)]=0$ 
for $t\in I$, and covariance $C(s,t)=E[U(s)U(t)]$ for all $s, t\in I$.
By Karhunen-Loeve expansion and Mercer's theorem, the covariance function $C(s,t)$ can be decomposed into
\begin{align}
\label{eq:Mercer Theorem - theoretical}
C(s,t)=\sum_{l=1}^\infty\lambda_l\psi_l(s)\psi_l(t) 
\end{align}
and the stochastic part $U(t)$ can be decomposed into
$$U(t)=\sum_{l=1}^\infty\xi_{l}\psi_l(t),$$
where $\psi_l(t)$ are the eigenfunctions of the covariance operator associated with $C(s,t)$, and $\lambda_1>\lambda_2>...$ are the the corresponding eigenvalues.
Moreover, the variance of the scores $\xi$ satisfies $var(\xi_l)=\lambda_l$.
Note that the functional principal components $\{\psi_l(t)\}$ (FPCs) form an orthonormal basis for $L^2 (I)$.

Let $X_1(t),...,X_n(t)$ represent $n$ iid copies of $X(t)$ with $t\in I$. Assume observations of $X_1(t),...,X_n(t)$ are made at discrete time points $t_{i1},...,t_{iN_i}$ for $i=1,...,n$, which are perturbed by independent random errors $\epsilon$.  
Specifically, let $Y_{ij}$ denote the random variable for the $j$th time point for subject $i$ with $j=1,...,N_i$ and $i=1,...,n$, with $N_i$ the number of measurements for subject $i$.
The model for $Y_{ij}$ is therefore given by,
\begin{align} \label{eq:model}
Y_{ij}
=X_i(t_{ij})+\epsilon_{ij}
=\mu(t_{ij})+U_i(t_{ij})+\epsilon_{ij}
=\mu(t_{ij})+\sum_{l=1}^\infty\xi_{il}\psi_l(t_{ij})+\epsilon_{ij}
\end{align}
where $\epsilon_{ij}$ represents an independent random measurement error term, distributed according to an exponential family with mean zero and variance $\sigma^2$,
which implies $\epsilon_{ij}$ are independent across both $i$ and $j$.
Moreover, we assume that $\epsilon_{ij}$ is independent of $U_i(t)$ (or equivalently $\xi_{il}$). 
Often, a Gaussian distribution is assumed for both error terms and stochastic parts, i.e. we have $\epsilon_{ij}\sim \mathcal N(0,\sigma^2)$ and $\xi_{il}\sim \mathcal N(0, \lambda_l)$. 

Not all $Y_{ij}$ are observed due to the presence of a DL. 
Let $\delta_{ij}$ denote the missingness indicator, i.e.  $\delta_{ij}=0$ if $Y_{ij}$ is observed, 
and $\delta_{ij}=1$ if $Y_{ij}$ is unobserved.
When $\delta_{ij}=1$,
the unobserved $Y_{ij}$ is assumed to have a value less than or equal to a specific threshold $c_{ij}$. 
For simplicity of notation, 
we assume the threshold is fixed i.e. $c_{ij}=c$ for all $i,j$.
Thus, the observed data consists of 
$$\{(t_{ij}, y_{ij}, \delta_{ij})\},\ i=1,...,n,\ j=1,...,N_i,$$ 
where $y_{ij}$ is missing for $\delta_{ij}=1$.
The presence of a DL complicates the likelihood function, as the contributions of the observations subject to the DL are represented by the cumulative probability distribution function rather than the probability density function. 
This complexity makes the likelihood function difficult to maximise.
To address this, Liu and Houwing-Duistermaat (2022) proposed locally approximating the probability distributions using either a linear function or a constant.
This approach results in computationally efficient estimators for the mean function (see Liu and Houwing-Duistermaat 2022) and covariance function (see Liu and Houwing-Duistermaat 2023).  
Their simulations showed that the local-linear estimator performed only slightly better than the local-constant estimator, but was less computationally efficient. 
Therefore, in this paper we will only consider the local-constant estimator. In the next section we will briefly summarize the estimation procedure for the mean and covariance function. Then we will propose a novel approach for FPCA using data subject to detection limits.

\subsection{Locally kernel weighted log-likelihood estimator for mean  and covariance function}

In this subsection, for simplicity and without loss of generalisability, we assume that $U_i(t)=0$ in formula (\ref{eq:model}).
The loglikelihood function, locally approximated by a constant, is given as follows (see Liu and Houwing-Duistermaat 2022)
\begin{footnotesize}
\begin{align} \label{eq:likelihood - local constant - approximation}
\begin{split}
L(\boldsymbol{\beta}; h, t)
=\sum_{i=1}^nw_i\sum_{j=1}^{N_i}&
\left[-0.251\delta_{ij} \left(\frac{c-\beta_0}{\sigma}\right)^2
+0.8194\delta_{ij} \frac{c-\beta_0}{\sigma}
\right.\\
& \left.-(0.5-0.5\delta_{ij}) \left(\frac{y_{ij}-\beta_0}{\sigma}\right)^2
\right]K_h(t_{ij}-t).
\end{split}
\end{align}
\end{footnotesize}
where $K_h(\cdot)=\frac{1}{h}K\left(\frac{\cdot}{h}\right)$ 
and $K(\cdot)$ is a kernel function, 
 and $w_i$ are weights. 
Two types of weights $w_i$ are considered, namely
$$w_{i}^{SUBJ}=\frac{1}{nN_i}$$
and
$$w_{i}^{OBS}=\frac{1}{\sum_{i=1}^nN_i}.$$
The SUBJ weight assigns the same weight to each subject and
the OBS weight which assigns the same weight to each observation. The latter assigns more weight to subjects with more observations.

Using loglikelihood function (\ref{eq:likelihood - local constant - approximation}), Liu and Houwing-Duistermaat (2022) derived the local constant estimator of the mean function as follows
\begin{align}\label{eq:likelihood - local constant - approximation-estimator}
\hat\mu^{LC}(t)=\hat\beta_0=\frac{R_0}{S_0},
\end{align}
where
$$S_0=\sum_{i=1}^nw_i\sum_{j=1}^{N_i}(1-0.498\delta_{ij})
K_h(t_{ij}-t)$$
and
$$R_0=\sum_{i=1}^nw_i\sum_{j=1}^{N_i}\left[-0.8194\delta_{ij}\sigma
+0.502\delta_{ij}c+(1-\delta_{ij})y_{ij}
\right] 
K_h(t_{ij}-t).$$

Note that Liu and Houwing-Duistermaat (2022)  only derived the asymptotic distribution for the local linear estimator (see Theorem 1 of their paper). The asymptotic distribution of $\hat\mu^{LC}(t)$ given in (\ref{eq:likelihood - local constant - approximation-estimator}) can be derived similarly.

For the rest of this section, we assume $\mu(t)=0$ and $\epsilon \sim N(0, \sigma^2)$ 
(for simplicity of notations). 
We propose the following fast local constant kernel weighted estimator of $C(s,t)$:
\begin{align} \label{eq:likelihood - C - approximation}
\tilde C(s,t)=\frac{R_{00}}{S_{00}}
\end{align}
where
$$R_{00}=\sum_{i=1}^n v_i\sum_{1\leq j\ne l\leq N_i}K_h(t_{ij}-s)K_h(t_{il}-t)C_{ijl}$$
and
$$S_{00}=\sum_{i=1}^n v_i\sum_{1\leq j\ne l\leq N_i}K_h(t_{ij}-s)K_h(t_{il}-t)D_{ijl}$$
with 
$$C_{ijl}=[-0.8194\delta_{ij}\sigma+0.502\delta_{ij}c+(1-\delta_{ij})y_{ij}]\cdot
[-0.8194\delta_{il}\sigma+0.502\delta_{il}c+(1-\delta_{il})y_{il}]$$
and
$$D_{ijl}=(1-0.498\delta_{ij})(1-0.498\delta_{il})$$
and $v_i$
$$v^{OBS}_i=\frac{1}{\sum N_i(N_i-1)}\quad \text{or}\quad  v^{SUBJ}_i=\frac{1}{nN_i(N_i-1)}.$$

Liu and Houwing-Duistermaat (2023) derived the asymptotic distribution of the local constant estimator given in 
equation (\ref{eq:likelihood - C - approximation}) and we summarise the result as the following lemma (and related notations can be found in their paper):
\begin{lemma} \label{lemma - covariance}
Under certain assumptions of kernel function, sampling density and continuity of covariance, 
for a fixed interior point $(s,t)\in I\times I$, we have,
$$(\Gamma_{n,N_i})^{-1/2}\left[\tilde C(s,t)-C(s,t)-B(s,t)\sigma^2
-\frac{h^2}{2}\sigma_K^2D(s,t)+o(h^2)\right]\to \mathcal{N}(0,1).$$
\end{lemma}

\subsection{Eigenvalues and functional principal components} 
The eigenvalues $\lambda_l$ and eigenfunctions $\psi_l(t)$ are
estimated from 
\begin{align}
\label{eq:adjusted - C - approximation}
\mathbf {\hat C}(\hat\psi_l(t))=\int\hat C(s,t)\hat\psi_l(s)ds=\hat\lambda_l\hat\psi_l(t)
\end{align}
where $\hat C(s,t)=\tilde C(s,t)-B(s,t)\sigma^2$ is the adjusted estimated covariance estimator derived from formula (\ref{eq:likelihood - C - approximation})
and $\mathbf{\hat C}$ is the corresponding covariance operator induced by $\hat C(s,t)$.

\begin{theorem}
Let $\hat\lambda_l,\ \hat\psi_l$ be defined by formula (\ref{eq:adjusted - C - approximation}), and suppose that 
the assumptions of Lemma \ref{lemma - covariance} hold. 
Moreover, assume that $\lambda_l\geq\lambda_{l+1}$.
For a given threshold $c_\lambda>0$, let $J\in \mathbb N$ denote the (finite) set of indices such that $\lambda_l>c_\lambda$ and
$\lambda_l>\lambda_{l+1},\ (l\in J)$. 
Then for each $l\in J$, we have
$$|\hat\lambda_l-\lambda_l|=O_p\left(T_{n,N_i}^{-1/2}\right)$$
$$||\hat\psi_l-\psi_l||=O_p\left(T_{n,N_i}^{-1/2}\right)$$
and 
$$\sup_{t\in I}|\hat\psi_l(t)-\psi_l(t)|=O_p\left(T_{n,N_i}^{-1/2}\right).$$
\end{theorem}
This result can be obtained using a similar approach as in the proof of Theorem 1, Lemma 1 and Theorem 2 in Beran and Liu (2014).

\subsection{Functional principal component scores}
We assume in this section that the mean function $\mu(t)=0$ 
and the functional principal components $\psi_l(t)$ are known.
To estimate the first $L$ scores $\xi_{i1}, ..., \xi_{iL}$ for each individual $i=1, ..., n$, we propose maximum likelihood estimation. Specifically,
the likelihood function for the $i$th person is given by
\begin{align}
L_i(\xi_{i1}, ..., \xi_{iL})
=\prod_{j=1}^{N_i}\left\{
\phi\left(y_{ij}; \sum_{l=1}^L\xi_{il}\psi_l(t_{ij}), \sigma^2\right)^{1-\delta_{ij}}
\Phi\left(c; \sum_{l=1}^L\xi_{il}\psi_l(t_{ij}), \sigma^2\right)^{\delta_{ij}}\right\}
\end{align}
Using the approximation technique as that in Liu and Duistermaat (2022), $\log(\Phi(x)), x\in[-1, 2]$ is approximated by
$$\log(\Phi(x))\approx-0.7127+0.8194x-0.251x^2, x\in[-1,2].$$
Therefore the loglikelihood function can be written as (ignoring a constant term)
\begin{align*}
\begin{split}%\label{likelihood}
l_i(\xi_{i1}, ..., \xi_{iL})
=\sum_{j=1}^{N_i}&
\left[-0.251\delta_{ij} \left(\frac{c-\sum_{l=1}^L\xi_{il}\psi_l(t_{ij})}{\sigma}\right)^2
\right.\\
& \left.
+0.8194\delta_{ij} \frac{c-\sum_{l=1}^L\xi_{il}\psi_l(t_{ij})}{\sigma}
\right.\\
& \left.-(0.5-0.5\delta_{ij}) \left(\frac{y_{ij}-\sum_{l=1}^L\xi_{il}\psi_l(t_{ij})}{\sigma}\right)^2
\right].
\end{split}
\end{align*}
Taking the first derivative of the loglikelihood function with respect to $\xi_{il}, l=1,...,L$ and set them to zero,
\begin{align*}
\begin{split}
\frac{\partial l_i}{\partial\xi_{il}}
=-\sigma^{-2}\sum_{j=1}^{N_i}&
\left[-0.502\delta_{ij} \left(c-\sum_{k=1}^L\xi_{ik}\psi_k(t_{ij})\right)\psi_l(t_{ij})
\right.\\
& \left.+0.8194\delta_{ij}\psi_l(t_{ij})\sigma
\right.\\
& \left.-(1-\delta_{ij}) \left(y_{ij}-\sum_{k=1}^L\xi_{ik}\psi_k(t_{ij})\right)\psi_l(t_{ij})
\right]=0
\end{split}
\end{align*}
which is equivalent to 
\begin{align*}
0=
&\sum_{j=1}^{N_i}(1-\delta_{ij})y_{ij}\psi_l(t_{ij})\\
&+(0.502c-0.8194\sigma)\sum_{j=1}^{N_i}\delta_{ij}\psi_l(t_{ij})\\
&-\sum_{k=1}^L\xi_{ik}\sum_{j=1}^{N_i}(1-0.498\delta_{ij})\psi_k(t_{ij})\psi_l(t_{ij})\\
=:&R_{il}+S_{il}-\sum_{k=1}^L\xi_{ik}T_{ikl}
\end{align*}
with $R_{il}=\sum_{j=1}^{N_i}(1-\delta_{ij})y_{ij}\psi_l(t_{ij})$,
$S_{il}=(0.502c-0.8194\sigma)\sum_{j=1}^{N_i}\delta_{ij}\psi_l(t_{ij})$ 
and
$T_{ikl}=\sum_{j=1}^{N_i}(1-0.498\delta_{ij})\psi_k(t_{ij})\psi_l(t_{ij})$.
These $L$ equations can be written as the following matrix form:
\begin{align*}
\begin{bmatrix}
    T_{i11}       & T_{i21}  & \dots & T_{iL1} \\
    T_{i12}       & T_{i22}  & \dots & T_{iL2} \\
    \vdots & \vdots & \vdots & \vdots \\
    T_{i1L}       & T_{i2L}  & \dots & T_{iLL} \\
\end{bmatrix}
\begin{bmatrix}
    \xi_{i1}       \\
    \xi_{i2}       \\
    \vdots \\
    \xi_{iL}       
\end{bmatrix}
=
\begin{bmatrix}
    R_{i1}+S_{i1}       \\
    R_{i2}+S_{i2}       \\
    \vdots \\
    R_{iL}+S_{iL}       
\end{bmatrix}.
\end{align*}
Therefore, the estimate of the scores can be obtained as
\begin{align*}
\begin{bmatrix}
    \hat\xi_{i1}       \\
    \hat\xi_{i2}       \\
    \vdots \\
    \hat\xi_{iL}       
\end{bmatrix}
=
\begin{bmatrix}
    T_{i11}       & T_{i21}  & \dots & T_{iL1} \\
    T_{i12}       & T_{i22}  & \dots & T_{iL2} \\
    \vdots & \vdots & \vdots & \vdots \\
    T_{i1L}       & T_{i2L}  & \dots & T_{iLL} \\
\end{bmatrix}^{-1}
\begin{bmatrix}
    R_{i1}+S_{i1}       \\
    R_{i2}+S_{i2}       \\
    \vdots \\
    R_{iL}+S_{iL}       
\end{bmatrix}.
\end{align*}
By writing
\begin{align*}
\begin{bmatrix}
    T_{i11}       & T_{i21}  & \dots & T_{iL1} \\
    T_{i12}       & T_{i22}  & \dots & T_{iL2} \\
    \vdots & \vdots & \vdots & \vdots \\
    T_{i1L}       & T_{i2L}  & \dots & T_{iLL} \\
\end{bmatrix}^{-1}
=:
\begin{bmatrix}
    Q_{i11}       & Q_{i21}  & \dots & Q_{iL1} \\
    Q_{i12}       & Q_{i22}  & \dots & Q_{iL2} \\
    \vdots & \vdots & \vdots & \vdots \\
    Q_{i1L}       & Q_{i2L}  & \dots & Q_{iLL} \\
\end{bmatrix},
\end{align*}
Therefore
\begin{align}
\label{eq: estimated scores}
\hat\xi_{il}=\sum_{k=1}^LQ_{ikl}(R_{ik}+S_{ik}), l=1,...,L.
\end{align}
\begin{remark}
Notice that, if we assume $\delta_{ij}$ is the discrete version of
an underlying missing mechanism denoted by a deterministic integrable function $\delta(t)$ and $\epsilon_{ij}$ is the discrete version of
an underlying measurement error mechanism denoted by a deterministic integrable function $\epsilon(t)$, as $N_i\to\infty$, we have
\begin{align*}
\frac{1}{N_i}R_{il}
&=\frac{1}{N_i}\sum_{j=1}^{N_i}(1-\delta_{ij})y_{ij}\psi_l(t_{ij})\\
&=\frac{1}{N_i}\sum_{j=1}^{N_i}(1-\delta_{ij}))
\left(\sum_{k=1}^L\xi_{ik}\psi_{k}(t_{ij})+\epsilon_{ij}\right)\psi_l(t_{ij})\\
&\to\sum_{k=1}^L\xi_{ik}\int(1-\delta(t))\psi_k(t)\psi_l(t)dt
+\int(1-\delta(t))\epsilon(t)\psi_l(t)dt,
\end{align*}
\begin{align*}
\frac{1}{N_i}S_{il}=
\frac{1}{N_i}
(0.502c-0.8194\sigma)\sum_{j=1}^{N_i}\delta_{ij}\psi_l(t_{ij})
\to(0.502c-0.8194\sigma)\int\delta(t)\psi_l(t)dt,
\end{align*}
\begin{align*}
\frac{1}{N_i}T_{ikl}
=\frac{1}{N_i}\sum_{j=1}^{N_i}(1-0.498\delta_{ij})\psi_k(t_{ij})\psi_l(t_{ij})
\to\int (1-0.498\delta(t))\psi_k(t)\psi_l(t)dt,
\end{align*}
which implies,
\begin{align*}
R_{il}=O(N_i), \
S_{il}=O(N_i), \
T_{ikl}=O(N_i),
\end{align*}
and therefore
$$Q_{ikl}=O(N_i^{-1}). $$

\end{remark}

Define
$$\zeta_i=\sqrt{N_i}(\hat\xi_{i1}-A_{i1}, ..., \hat\xi_{iL}-A_{iL})^T,\ i=1, ..., n$$
with
\begin{align}
\label{eq:A}
A_{il}=\lim_{N_i\to\infty}%A_{il,N_i}=
\sum_{k=1}^LQ_{ikl}\sum_{j=1}^{N_i}(1-\delta_{ij})\sum_{k^\prime=1}^L\xi_{ik^\prime}\psi_{k^\prime}(t_{ij})\psi_k(t_{ij})+\sum_{k=1}^LQ_{ikl}S_{ik}.
\end{align}

\begin{theorem}
Under the Gaussian assumption, for $i=1,...,n$, as $N_i\to\infty$, we have
$$\zeta_i\to\mathcal N(0, V)$$
where
$$V=[v_{ll^\prime}]_{l, l^\prime=1,...,L}$$
with 
$$v_{ll^\prime}
=\lim_{N_i\to\infty}\sum_{k, k^\prime=1}^LQ_{ikl}Q_{ik^\prime l^\prime}
N_i\sum_{j=1}^{N_i}(1-\delta_{ij})^2\psi_k(t_{ij})\psi_{k^\prime}(t_{ij})\sigma^2$$
\end{theorem}
\begin{proof}
For the expectation of the $\hat\xi_{il}$, by noticing $E[y_{ij}]=0$,
we have
\begin{align*}
E[\hat\xi_{il}]
&=\sum_{k=1}^LQ_{ikl}(E[R_{ik}]+S_{ik})\\
&=\sum_{k=1}^LQ_{ikl}\sum_{j=1}^{N_i}(1-\delta_{ij})E[y_{ij}]\psi_k(t_{ij})+\sum_{k=1}^LQ_{ikl}S_{ik}\\
&=\sum_{k=1}^LQ_{ikl}\sum_{j=1}^{N_i}(1-\delta_{ij})\sum_{k^\prime=1}^L\xi_{ik^\prime}\psi_{k^\prime}(t_{ij})\psi_k(t_{ij})+\sum_{k=1}^LQ_{ikl}S_{ik}\\to A_{il}.
% &=O(1).
\end{align*}
For the conditional covariance $var(\hat\xi_{il}|\xi_{il})$, given $\xi_{il}$, $var(\sum_{k^\prime=1}^L\xi_{ik^\prime}\psi_{k^\prime}(t_{ij})|\xi_{il})=0,$
therefore,
\begin{align*}
var(\sqrt{N_i}\hat\xi_{il}|\xi_{il})
&=N_ivar\left(\sum_{k=1}^LQ_{ikl}R_{ik}|\xi_{il}\right)\\
&=N_ivar\left(\sum_{k=1}^LQ_{ikl}\sum_{j=1}^{N_i}(1-\delta_{ij})y_{ij}\psi_k(t_{ij})|\xi_{il}\right)\\
&=N_ivar\left(\sum_{k=1}^LQ_{ikl}\sum_{j=1}^{N_i}(1-\delta_{ij})\left(\sum_{k^\prime=1}^L\xi_{ik^\prime}\psi_{k^\prime}(t_{ij})+\epsilon_{ij}\right)\psi_k(t_{ij})|\xi_{il}\right)\\
&=N_i\sum_{j=1}^{N_i}(1-\delta_{ij})^2
\sum_{k=1}^LQ^2_{ikl}\psi^2_k(t_{ij})
\left(0
+var(\epsilon_{ij})\right)\\
&=N_i\sum_{j=1}^{N_i}(1-\delta_{ij})^2
\sum_{k=1}^LQ^2_{ikl}\psi^2_k(t_{ij})
\left(
\sigma^2\right)\to v_{ll},
\end{align*}
and, for $l\ne l^\prime$,
\begin{align*}
&cov(\sqrt{N_i}\hat\xi_{il}, \sqrt{N_i}\hat\xi_{il^\prime}|\xi_{il}, \xi_{il^\prime})\\
&=N_icov\left(\sum_{k=1}^LQ_{ikl}(R_{ik}+S_{ik}), \sum_{k=1}^LQ_{ik^\prime l^\prime}(R_{ik^\prime}+S_{ik^\prime})|\xi_{il}, \xi_{il^\prime}\right)\\
&=N_icov\left(\sum_{k=1}^LQ_{ikl}R_{ik}, \sum_{k=1}^LQ_{ik^\prime l^\prime}R_{ik^\prime}|\xi_{il}, \xi_{il^\prime}\right)\\
&=N_icov(\sum_{k=1}^LQ_{ikl}\sum_{j=1}^{N_i}(1-\delta_{ij})
\left(\sum_{m=1}^L\xi_{im}\psi_{m}(t_{ij})+
\epsilon_{ij}\right)\psi_k(t_{ij}), \\
&\quad\quad\quad\sum_{k^\prime=1}^LQ_{ik^\prime l^\prime}\sum_{j^\prime=1}^{N_i}(1-\delta_{ij^\prime})\left(\sum_{m^\prime=1}^L\xi_{im^\prime}\psi_{m^\prime}(t_{ij^\prime})+\epsilon_{ij^\prime}\right)\psi_{k^\prime}(t_{ij^\prime})\\
&\quad\quad\quad |\xi_{il}, \xi_{il^\prime})\\
&=N_icov\left(\sum_{k=1}^LQ_{ikl}\sum_{j=1}^{N_i}(1-\delta_{ij})
\epsilon_{ij}\psi_k(t_{ij}), 
\sum_{k^\prime=1}^LQ_{ik^\prime l^\prime}\sum_{j^\prime=1}^{N_i}(1-\delta_{ij^\prime})\epsilon_{ij^\prime}\psi_{k^\prime}(t_{ij^\prime})\right)\\
&=\sum_{k, k^\prime=1}^LQ_{ikl}Q_{ik^\prime l^\prime}
N_i\sum_{j=1}^{N_i}(1-\delta_{ij})^2\psi_k(t_{ij})\psi_{k^\prime}(t_{ij})\sigma^2\\
&\to v_{ll^\prime}.
\end{align*}

\end{proof}
\begin{remark}
For the special case that $\delta_{ij}=0$, that is, none of the observations
are subjected to detection limit, 
and since as $N_i\to\infty$, $\frac{1}{N_i}\sum_{j=1}^{N_i}\psi_{k}(t_{ij})\psi_{k^\prime}(t_{ij})\to0$ if $k\ne k^\prime$ and $\frac{1}{N_i}\sum_{j=1}^{N_i}\psi_{k}^2(t_{ij})\to1$, 
we have 
$R_{il}=\sum_{j=1}^{N_i}y_{ij}\psi_l(t_{ij})$, 
$S_{il}=0$, 
$\frac{1}{N_i}T_{ikl}\to1$ if $k=l$, 
$\frac{1}{N_i}T_{ikl}\to0$ if $k\ne l$.
Then $E\left[\hat\xi_{il}\right]\to\xi_{il}$, $var\left(\sqrt{N_i}\hat\xi_{il}\right)
\to\sigma^2$
and $cov\left(\sqrt{N_i}\hat\xi_{il}, \sqrt{N_i}\hat\xi_{il^\prime}\right)=0$ if $l\ne l^\prime$.
\end{remark}

\begin{remark}
For the special case that $\delta_{ij}=1$ that is all of the observations
are subject to detection limit,
and since as $N_i\to\infty$, $\frac{1}{N_i}\sum_{j=1}^{N_i}\psi_{k}(t_{ij})\psi_{k^\prime}(t_{ij})\to0$ if $k\ne k^\prime$ and $\frac{1}{N_i}\sum_{j=1}^{N_i}\psi_{k}^2(t_{ij})\to1$, 
we have 
$R_{il}=0$, $S_{il}=(0.502c-0.8194\sigma)\sum_{j=1}^{N_i}\psi_l(t_{ij})$, $\frac{1}{N_i}T_{ikl}\to0.502$ if $k=l$, $\frac{1}{N_i}T_{ikl}\to0$ if $k\ne l$.
Then $E\left[\hat\xi_{il}\right]\to\left(c-\frac{0.8194}{0.502}\sigma\right)\int\psi_l(t)dt$, 
$var\left(\sqrt{N_i}\hat\xi_{il}\right)=0$
and $cov\left(\sqrt{N_i}\hat\xi_{il}, \sqrt{N_i}\hat\xi_{il^\prime}\right)=0$ if $l\ne l^\prime$.
\end{remark}

\begin{remark}
In the general case that $\psi_l(t)$ are unknown and should be estimated say using formula (\ref{eq:adjusted - C - approximation}),
the $\psi_l(t)$ in 
formula (\ref{eq: estimated scores}) will be replaced by $\hat\psi_l(t)$. 
The proof of asymptotic property of the estimated scores is similar to that of the Theorem 1, except a term of $O_p\left(T_{n,N_i}^{-1/2}\right)$, see proof of Theorem 4 in Beran and Liu (2016).
\end{remark}

\section{Simulation study}
In this section, we evaluate the performance of our proposed estimators of $\psi_l(t),\ \xi_{il}, \ l\in\mathbb N$ via extensive simulations. 
We compare the performance with a standard method where the missing observations are replaced with the DL value (Yao et al. 2005), in terms of bias, efficiency, asymptotic behaviour and computation time.

We assume that $\mu(t)=0$ for simplicity and define the true zero mean random function $U(t)$ as follows
\begin{align} \label{eq:simulation}
 U(t)=\xi\psi(t),\ t\in[0,1],   
\end{align}
where $\psi(t)=\sqrt{2}\cos(4\pi t)$, $\xi\sim \mathcal N(0, \lambda)$ and $\lambda=2$.
The observed time points $t_{ij}\sim_{iid} \mathcal U[0, 1]$. 
Additive errors are sampled from $\epsilon_{ij}\sim \mathcal N(0, 1)$.
The scores $\xi_i\sim \mathcal N(0, 2)$.
Then, the data is generated by 
$$Y_{ij}=X_i(t_{ij})+\epsilon_{ij}=\xi_i\psi(t_{ij})+\epsilon_{ij},\ i=1,...,n, \ j=1,...,N_i.$$
Finally, we create missing data by  replacing observations smaller than the DL with $c$ which is the value of the DL. We consider $c \in \{-1,0\}$.

\subsection{Simulation study on eigenfunctions}

We consider two settings, namely a sparse and a dense grid for the observations for each subject $i$. Specifically
\begin{itemize}
\item Sparse setting: $N_i\sim \mathcal U\{3, 4, 5, 6, 7, 8, 9, 10\}$ i.e. $N_i$ are iid from a discrete uniform distribution 
in $\{3, 4, ..., 10\}$.
\item Dense setting: $N_i\sim \mathcal U\{75, 76, ..., 100\}$ i.e. $N_i$ are iid from a discrete uniform distribution 
in $\{75, 76, ..., 100\}$.
\end{itemize}
For both settings we will generate one dataset of $n=100$ subjects.

For the sparse setting, the upper panel of Figure \ref{fig:FigureObserveration_FPCA_2DL} depicts the data  without a DL (left plot) and subject to a DL of $c=0$ (middle plot) and of $c=-1$ (right plot). The proportion of observations subject to DL of $c=0$ and of $c=-1$ is 49.46\% and 25.66\%, respectively. 
For the dense setting, the lower panel of Figure \ref{fig:FigureObserveration_FPCA_2DL} shows the data  without a DL(left plot) and subject to a DL of $c=0$ (middle plot) and of $c=-1$ (right plot). The proportion of observations subject to DL of $c=0$ and of $c=-1$ is 50.06\% and 27.26\% respectively.

\begin{figure}[htp]
\centering
    \includegraphics[width=\textwidth]{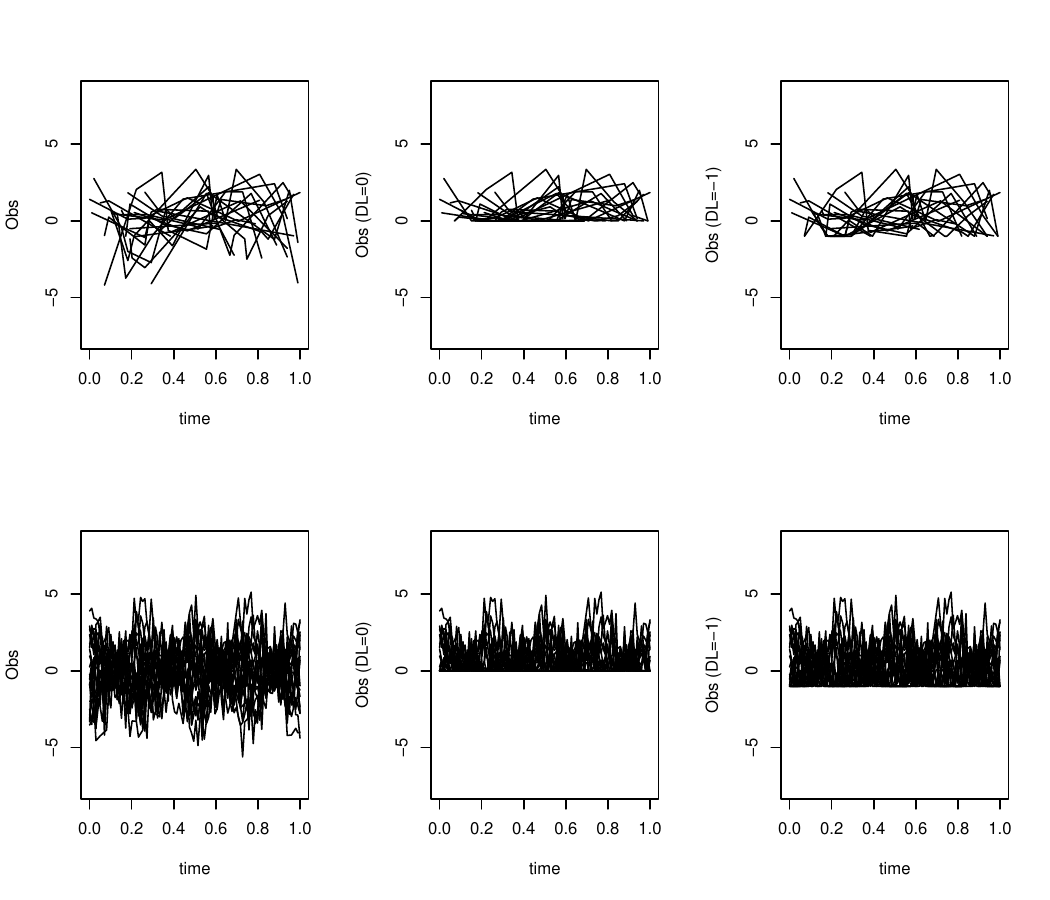}
    \caption{The first 20 of the 100 trajectories in the first replicate of the sparse (upper panel) and dense (lower panel) setting. In each panel, left: data without a DL, middle: data with a DL of 0, right: data with a DL of $-1$.}
    \label{fig:FigureObserveration_FPCA_2DL}
\end{figure}

To estimate the covariance function, we consider the following methods:
\begin{itemize}
\item Local which is based on local constant approximations using the OBS weighting scheme (Liu and Duistermaat 2023).
\item PACE which does not  adjust for the detection limit  (Yao et al., 2005).
\end{itemize}
The covariance functions are estimated on 100 equal-distant time points in $[0,1]\times[0,1]$.
The variance of $\epsilon_{ij}$ is estimated as the mean squared error based on
the least-squared fit using all the data (including the values subject to DL). 
We use the Gaussian kernel for the estimation procedure. 
The bandwidth $h$ is selected based on the integrated squared error (ISE) computed for a dense grid of values, see details in Liu and Duistermaat (2023).
Once the covariance function is estimated, the eigenfunctions can be estimated via formula (\ref{eq:Mercer Theorem - theoretical}).

Figure \ref{fig:FigureEstimatedEigenfunction} shows the true and estimated eigenfunction using different methods under different sampling settings and different detection limit scenarios.
Apparently, for all sampling settings and detection limit scenarios, the estimated eigenfunction (red curves) using the proposed method is closer 
to the true eigenfunction (black curves) than that by using PACE method under the default settings (blue curves).
Moreover, the proposed method is much faster than PACE, for example, for the dense sampling setting and DL$=0$ scenario, the computation time for the proposed method is 0.046 seconds and that for the PACE method is 791 seconds.
\begin{figure}[htp]
\centering
    \includegraphics[width=\textwidth]{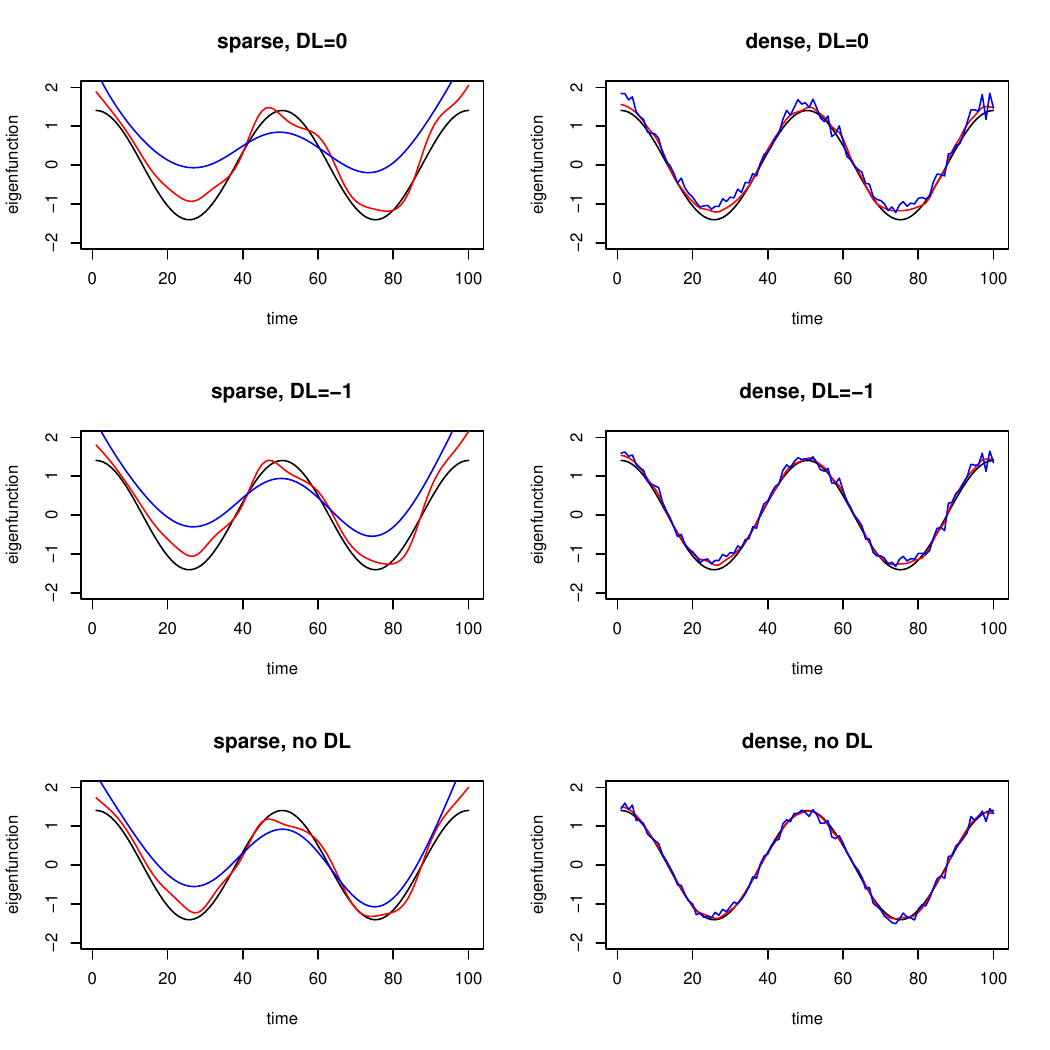}
    \caption{Estimated eigenfunctions. Upper left is the sparse setting and DL$=0$ scenario where black curve is the true eigenfunction, red curve is the corresponding estimated eigenfunction using the proposed method, and blue curve is the corresponding estimated eigenfunciton using the PACE method. Upper right is the dense setting and DL$=0$ scenario. Middle left is the sparse setting and DL$=-1$ scenario. Middle right is the dense setting and DL$=-1$ scenario. Lower left is the sparse setting without DL scenario. Lower right is the dense setting without DL scenario.}
    \label{fig:FigureEstimatedEigenfunction}
\end{figure}
Table 1 gives the integrated squared error (ISE*1000) of estimated eigenfunction using different methods under different sampling settings for the two considered detection limits.
For PACE, the ISE of estimated eigenfunction is much larger than that of the local method proposed here for both dense and sparse settings and different detection limit scenarios, as could be expected since it does not take into account of the DL.
Moreover, the dense setting has smaller ISE compared to the sparse setting which is also expected as more data are used.
Comparing DL$=-1$ with DL$=0$, ISE is smaller for DL$=-1$. 
This can be explained by the fact that there is more information for DL$=-1$ scenario.

\begin{table}
\begin{center}
\begin{tabular}{ |c|cc| }  
 \hline
 & Local & PACE   \\
 \hline
sparse, DL$=0$  & 88 & 763  \\
 \hline
dense, DL$=0$  & 14 & 54 \\
 \hline
sparse, DL$=-1$  & 71 & 516  \\
 \hline
dense, DL$=-1$  & 6 & 22 \\
 \hline
sparse, no DL  & 47 & 281  \\
 \hline
dense, no DL  & 1 & 10 \\
 \hline
\end{tabular}
\label{tab:eigenfunction}
\caption{ISE*1000 of the estimated eigenfunction for the sparse and dense setting with DL$=0$ and DL$=-1$.
}
\end{center}
\end{table}

\subsection{Simulation study on scores}
To evaluate the performance of our method in estimating the scores, we generated 100 replicates each with 100 subjects. Further to verify Theorem 2, we consider increasing number of sampling points $N_i$. 
We consider a sparse and a dense grid for the observations for each subject $i$. Specifically
\begin{itemize}
\item Sparse setting: $N_i\sim \mathcal U\{3, 4, 5, 6, 7, 8, 9, 10\}*\frac{M}{100}$ i.e. $N_i$ are iid from a discrete uniform distribution 
in $\{3, 4, ..., 10\}*\frac{M}{100}$.
\item Dense setting: $N_i\sim \mathcal U\left\{\frac{3}{4}*M, \frac{3}{4}*M+1,..., M\right\}$ i.e. $N_i$ are iid from a discrete uniform distribution 
in $\left\{\frac{3}{4}*M, \frac{3}{4}*M+1,..., M\right\}$.
\end{itemize}
We will consider $M=100, 200, 500, 1000$.

In order to evaluate the accuracy of the proposed estimated scores,
we assume that the functional principal components $\psi(t)$ are known in (\ref{eq:simulation}).
Define the mean squared error ($MSE^*$) as 
$$MSE^*=\frac{1}{n}\sum_{i=1}^n(\hat\xi_{i}-\xi_i)^2$$
and another mean squared error ($MSE^{**}$) as 
$$MSE^{**}=\frac{1}{n}\sum_{i=1}^n(\hat\xi_{i}-A_i)^2,$$
where $A_i$ is defined in formula (\ref{eq:A}).
Notice that $MSE^{**}$ is for evaluation of the asymptotic result of Theorem 2.

The results are given in Table \ref{tab:scores}. 
The mean of the estimated scores for all 100 subjects and 100 replicates 
is given in the second column in Table \ref{tab:scores}. It appears  
that the values of the mean for the dense and sparse setting with the two DL values are close to the true mean  of zero.
The mean across the replicates of the estimates of the variance of the estimated scores is given in the third
column in Table \ref{tab:scores}. Here, the true value is equal to two. It appears that the variance
in the DL=0 scenario is less accurate than that in DL=$-1$ scenario, as more
information is missing in DL=0 scenario than in DL=$-1$ scenario.
The averaged $MSE^*$ and averaged $MSE^{**}$ are given in the fourth 
and fifth column respectively in Table \ref{tab:scores}.
It appears that when the number of observations increases, the values of these MSE's are decreasing for each setting and each DL scenario which is excepted under Theorem 2.
The $MSE^*$ and $MSE^{**}$ are smallest  
in the dense setting with DL=$-1$ which is the scenario where most information is available.
{\color{black} The average $MSE^{**}$ values are smaller than the average $MSE^*$ values across all settings and scenarios, except in the dense setting with a detection limit (DL) of $-1$. This is consistent with the expectation from Theorem 2, as $MSE^{**}$ quantifies the mean squared error between the estimated score and its asymptotic expectation. Notably, this expectation does not coincide with the true score when a detection limit is present.}
The last two columns in Table \ref{tab:scores} shows the mean of the estimates of the variance (variance$^1$) and of $MSE^*$ ($MSE^{*1}$) of the estimated scores using the traditional method, i.e. the numerical integral of observation and the estimated eigenfunction ($\hat\xi_i=\frac{1}{N_i}\sum_{j=1}^{N_i} Y_{ij}\hat\psi(t_{ij})$). It appears that when using the traditional method the variance is underestimated and the estimated scores are much less accurate than when using the  method in this paper. 

\begin{table}
    \centering
    \begin{tabular}{|c|c|c|c|c|c|c|}
 \hline
 \multicolumn{7}{|c|}{sparse (DL$=0$)}\\ \hline
 &  mean & variance &$MSE^*$ &$MSE^{**}$ & variance$^1$ & $MSE^{*1}$ \\ \hline
M=100 &  0.000 &1.839 &31.094 &15.989 & 0.771 & 75.890\\ \hline
M=200 & -0.009 &1.746 &16.092  &8.099 & 0.649 & 62.750\\ \hline
M=500 & -0.016 &1.664  &8.448  &3.229 & 0.564 & 54.837\\ \hline
M=1000& -0.017 &1.629  &6.167  &1.683 & 0.533 & 53.230\\ \hline
 \multicolumn{7}{|c|}{sparse (DL$=-1$)}\\ \hline
 &  mean & variance &$MSE^*$ &$MSE^{**}$ & variance$^1$ & $MSE^{*1}$ \\ \hline
M=100 &  -0.001 &2.119 &26.354 &18.233 & 1.242 & 43.464\\ \hline
M=200 &  -0.011 &2.015 &12.185  &9.703 & 1.130 & 30.925\\ \hline
M=500 &  -0.019 &1.916  &5.256  &4.726 & 1.050 & 23.460\\ \hline
M=1000&  -0.018 &1.875  &3.125  &3.231 & 1.011 & 21.874\\ \hline
 \multicolumn{7}{|c|}{dense (DL$=0$)}\\ \hline
 &  mean & variance &$MSE^*$ &$MSE^{**}$ & variance$^1$ & $MSE^{*1}$ \\ \hline
M=100 &   0.007 &1.609 &4.974 &0.893 & 0.516 & 49.918\\ \hline
M=200 &  -0.004 &1.625 &4.313 &0.452 & 0.514 & 50.525\\ \hline
M=500 &  -0.017 &1.618 &3.714 &0.180 & 0.507 & 50.241\\ \hline
M=1000&  -0.018 &1.610 &3.779 &0.088 & 0.506 & 50.593\\ \hline
 \multicolumn{7}{|c|}{dense (DL$=-1$)}\\ \hline
 &  mean & variance &$MSE^*$ &$MSE^{**}$ & variance$^1$ & $MSE^{*1}$ \\ \hline
M=100 &   0.004 &1.851 &2.232 &2.465 & 1.002 & 19.421\\ \hline
M=200 &  -0.006 &1.869 &1.620 &2.023 & 1.003 & 19.426\\ \hline
M=500 &  -0.019 &1.864 &1.110 &1.716 & 0.994 & 19.019\\ \hline
M=1000&  -0.020 &1.855 &1.083 &1.607 & 0.986 & 19.495\\ \hline
    \end{tabular}
    \caption{Asymptotic results for estimated scores for different settings under different DL.
    Second column: mean of the estimated scores from proposed method for all 100 subjects and 100 replicates.
    Third column: averaged variance of the estimated scores from proposed method.
    Fourth column: averaged $MSE^*$ from proposed method.
    Fifth column: averaged $MSE^{**}$ from proposed method.
    Sixth column (variance$^1$): averaged variance of the estimated scores from traditional method.
    Last column ($MSE^{*1}$): averaged $MSE^*$ of the estimated scores from traditional method.}
    \label{tab:scores}
\end{table}

\section{Data application} 
We illustrate our method by applying it to data on a biomarker from a longitudinal biomarker study involving scleroderma patients. Scleroderma is a complex and heterogeneous disease with variability in severity progression among patients. The study includes 217 patients who attended hospital visits between 2010 and 2015. Typically, scleroderma patients visit the hospital every six months to monitor disease progression. However, missed appointments and unrecorded data led to a sparse and unbalanced dataset. The data were collected under an ethically approved observational study protocol (HRA number 15/NE/0211).

For the biomarker aldose reductase (AR) with 7.8\% missing data due to DL, Liu and Houwing-Duistermaat (2022, 2023) estimated the mean and covariance functions using the local constant approximation method under the OBS scheme. Here, data on  90 patients with a total of 268 observations were used. 
Observations from time points lacking outcome or biomarker values, as well as certain outliers (e.g., instances where AR exceeded three times the standard deviation at a given time point, and patients with only a single observation) were excluded. 
The observed longitudinal profiles minus the mean function are displayed in Figure \ref{fig:FigureARminusMean_cov}. Note that the number of observations after 30 months is small hence estimates after 30 months are not reliable (see also Liu and Houwing-Duistermaat, 2022).
The estimate of the covariance function
is  shown in the Figure \ref{fig:FigureAREstimateSparse_cov}.

\begin{figure}[htp]
    \centering
    \includegraphics[width=\textwidth]{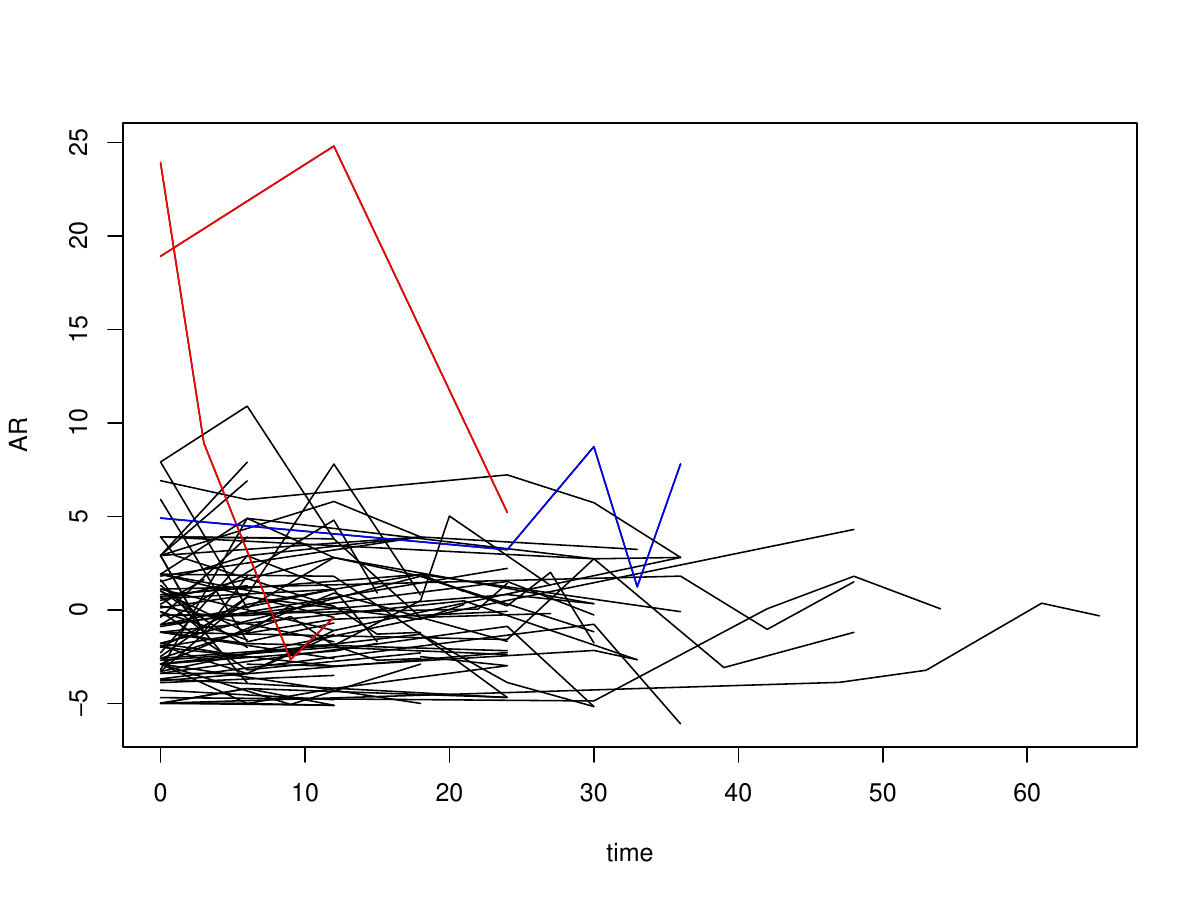}
    \caption{The AR observations with estimated mean being subtracted. 
    The red and blue curves correspond to the outliers in Figure \ref{fig:FigureAFEstimateScores.pdf}.}
    \label{fig:FigureARminusMean_cov}
\end{figure}

\begin{figure}[htp]
    \centering
    \includegraphics[width=\textwidth]{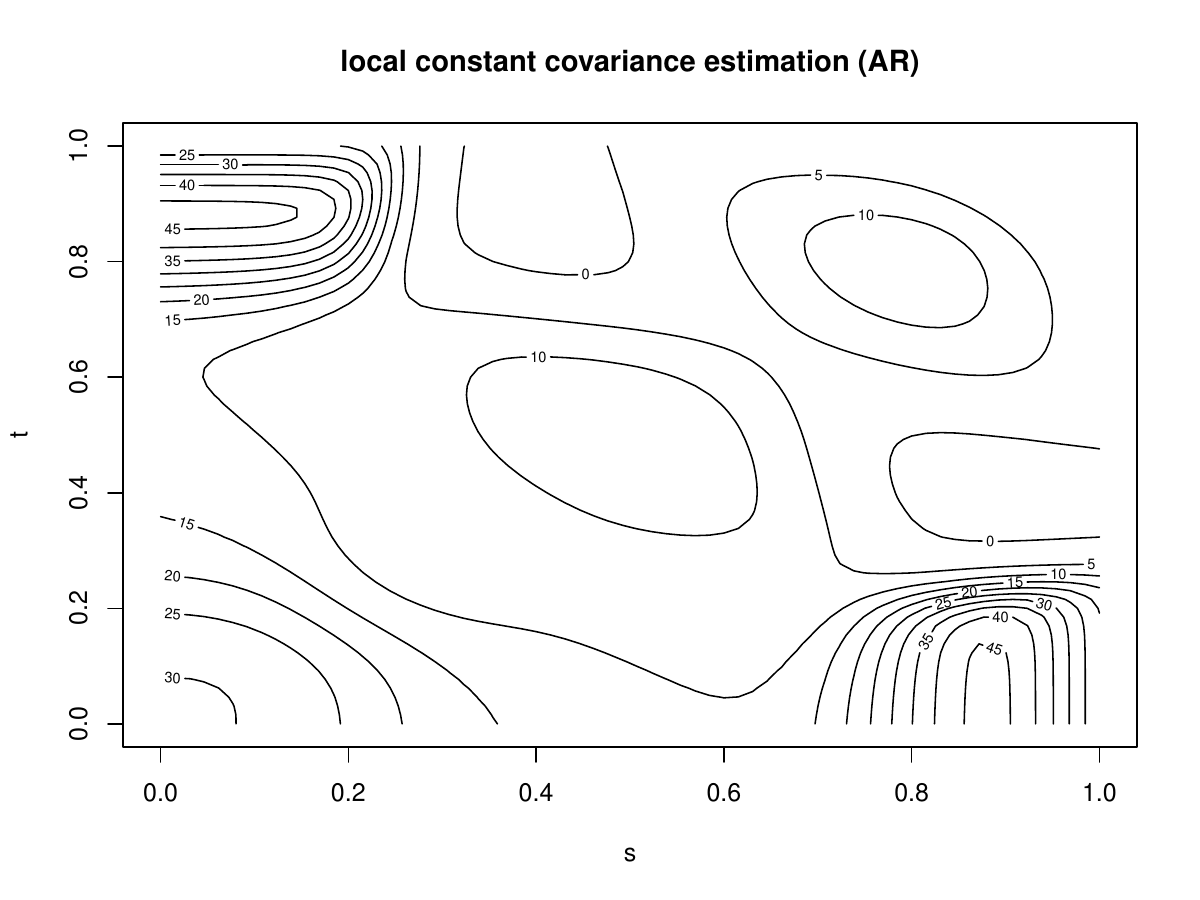}
    \caption{The covariance estimation for AR by using local constant estimation method.}
    \label{fig:FigureAREstimateSparse_cov}
\end{figure}

Using the proposed methods, we aim to estimate the FPCs and scores. We selected two FPCs which together explained 90\% of the variation, where the first and second FPCs explain 66.0\% and 25.4\% of the variation, respectively.
in Figure \ref{fig:FigureARPCAEst.pdf}, the two estimated FPCs are given (left plot).
Also the estimated mean function plus and minus estimated FPCs are given (middle and right plots).
In the left plot of Figure \ref{fig:FigureARPCAEst.pdf}, the red and blue curves are the estimated first and second FPCs respectively.
In the middle and right plots, the black curve is the estimated mean function and the red curves are the estimated mean function plus and minus the first and second FPC respectively.
It appears that the first FPC captures the variation mode in the first 15 months, while the second FPC captures the variation mode in the second 15 months (months 15-30).

\begin{figure}[htp]
    \centering
    \includegraphics[width=\textwidth]{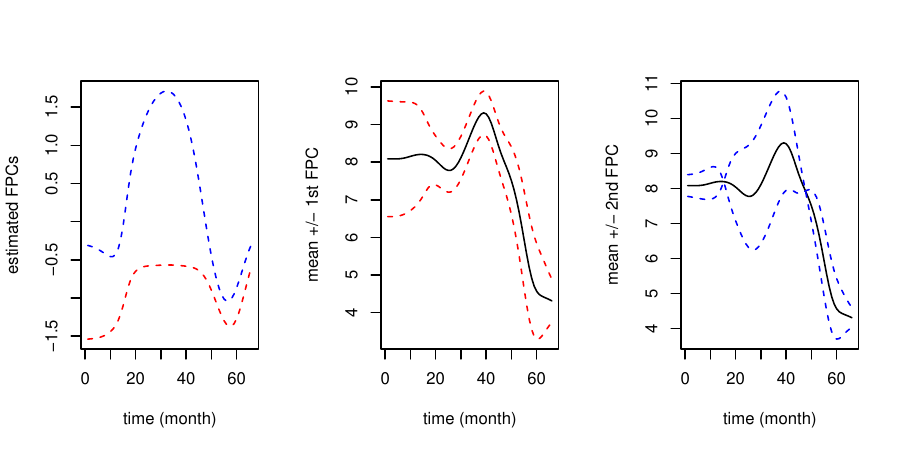}
    \caption{The estimated FPCs for AR by using the proposed method. 
    The red curve in the left plot is the estimated first FPC and the blue is the estimated second FPC.
    In the middle plot, the black curve is the estimated mean function, the red curves are the estimated mean function plus and minus the first FPC.
    In the right plot, the black curve is the estimated mean function, the red curves are the estimated mean function plus and minus the second FPC.}
    \label{fig:FigureARPCAEst.pdf}
\end{figure}

The estimated scores are shown in Figure \ref{fig:FigureAFEstimateScores.pdf}.
We observe two outliers from the first component which are colored red (subject 7, 67), and one outlier from the second component which is colored blue (subject 4).
The curves of these three individuals are also colored in Figure \ref{fig:FigureARminusMean_cov}.
The subjects colored red have large deviations from the mean in the first 15 months while the subject colored blue has a large deviation of the mean in the period 15 to 30 months.
\begin{figure}[htp]
    \centering
    \includegraphics[width=\textwidth]{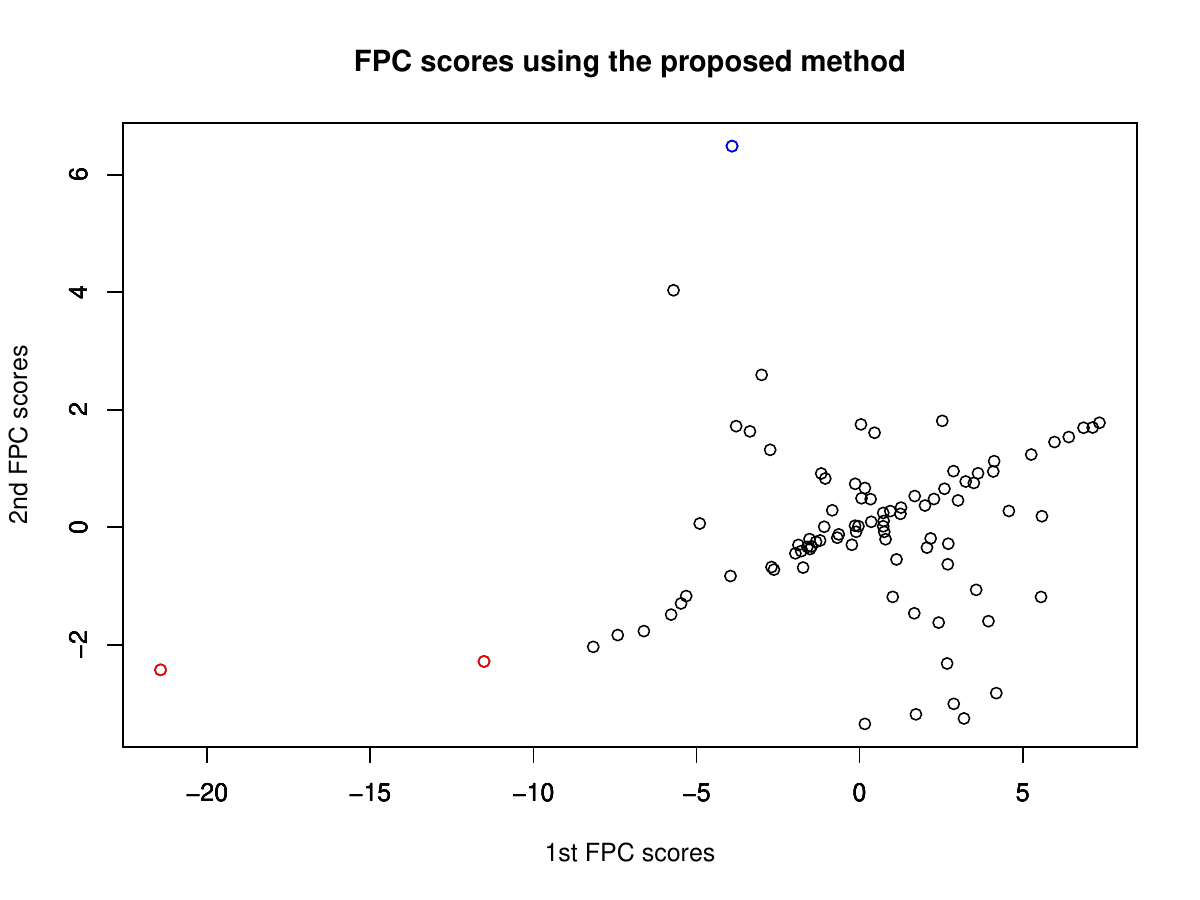}
    \caption{The estimated FPC scores for AR by using the proposed method. }
    \label{fig:FigureAFEstimateScores.pdf}
\end{figure}

To evaluate the performance of the reconstruction of the AR biomarkers based on the first two FPCs, the mean integrated square of errors (MISE)
$$IMSE=\frac{1}{n}\sum_{i=1}^n\int (X_i(t)-\hat X_i(t))^2dt,$$
were calculated. The MISE appears to be 11.85, while the MISE using the naive method ignoring DL is 13.71.

\section{Discussion}

We have proposed novel estimators for the functional principal components and scores for sparse and dense temporal data subject to a DL. 
The eigenfunctions are estimated using the estimated covariance function based on local smoothing using kernel functions (Liu and Houwing-Duistermaat, 2023). 
The scores are estimated using maximum likelihood estimation for each subject. 
We derived the asymptotic properties of the estimator for the scores. 
We investigated the small-sample properties of our estimators via simulations.
We compared our method to the one that ignores the presence of a DL in the data sample, i.e. replaces the missing value with the value of the DL. 
We showed that our methods performed better; the ISE of the estimated eigenfunctions and the MSE of the estimated scores are much smaller than those obtained using the standard method. Further we showed that when the sample size increases the MSE's of the estimated scores decrease.

We illustrated the method using data from a biomarker study. 
The first two components appeared to represent 90\% of the variation. 
The first component appeared to capture the variance in the first 15 months, while the second captures the variance of the curves in months 15 to 30. 
Investigation of the obtained scores identified three subjects with outlying trajectories, two in the first period of 15 months and one in the second period of 15 months.

The estimated eigenfunctions and scores can be used to assess the relationship between a functional predictor and an outcome following the work of Muller and Stadmuller (2005). 
Recently, Murphy et al (2024) proposed a generalised linear functional model with a functional predictor subject to detection limits. 
They estimate the mean and covariance functions using a full likelihood approach without the constant approximation. 
In our paper (Liu and Houwing-Duistermaat, 2023), we showed that our fast approximations in the likelihood function to estimate the covariance function perform well especially for dense observations and when there are not too many missing values (less than 30\%). 
In this paper we showed that the eigenfunctions estimated from this covariance function are accurate. 
Murphy et al (2024) did not evaluate the performance of their method in estimating the scores and eigenfunctions and just plugged in these estimates for the functional covariate in the linear predictor for an outcome variable. 
Further, there appeared to be not much difference in the prediction error of their estimator which takes the DL into account and the one that ignores the DL. 
Thus, it can be expected that plugging in the our estimated principal functional components will perform well for many situations. 

The scores are random variables, however, in this paper we did not estimate the variance of the scores and only the mean. 
In fact, since we did not take into account the randomness, we assume that the scores are fixed. 
For regression problems, taking into account the randomness might provide better estimates of the variance of the parameter functions. 
Further, for large $n$ plugging in the estimates of the eigenfunctions and scores might work well, but for smaller $n$ the measurement error might need to be taken into account. 
To take the variance of the scores and the measurement error into account, a joint modeling approach might be adopted. 
More research is needed here.

When a large number of components are used, the inverse of the $T$ matrix providing the $Q$ matrix in equation (\ref{eq: estimated scores}) might not be precise. 
In this case, a sequential approach might be adopted. 

Finally, in many studies with temporal data, we have to deal with data missing at random, where subjects drop out of the study. 
Indeed, this might also be the case in our biomarker study, where we have only a few patients with a follow-up time of more than 30 months. 
Such missing data mechanisms need other approaches for example by using random effects. 
This is a topic of future research.  

\section{Acknowledgement}
The authors would like to acknowledge the contribution of the COST Action CA21169, supported by COST (European Cooperation in Science and Technology)
and the support of the London Mathematical Society (EN-2324-07).

\end{document}